\documentclass[reqno,centertags,11pt]{amsart}
\usepackage{amsmath,amsthm,amscd,amssymb} \usepackage{latexsym}
\usepackage{graphicx}
\usepackage{color}

 \newcommand{\R}{{\mathbb{R}}}
 \newcommand{\C}{{\mathbb{C}}}
 
 \newcommand{\Z}{{\mathbb{Z}}}

 {

\newcommand{\beq}{\begin{equation}}
\newcommand{\eeq}{\end{equation}}
\newcommand{\bdm}{\begin{displaymath}}
\newcommand{\edm}{\end{displaymath}} \newcommand{\ba}{\begin{align}}
\newcommand{\ea}{\end{align}} \newcommand{\bpf}{\begin{proof}}
\newcommand{\epf}{\end{proof}}

 \allowdisplaybreaks

\DeclareMathOperator{\Tr}{Tr}

\newtheorem{theorem}{Theorem}[section]
\newtheorem{lemma}[theorem]{Lemma}
\newtheorem{corollary}[theorem]{Corollary}

\theoremstyle{definition}
\newtheorem{definition}[theorem]{Definition}
\theoremstyle{remark}
\newtheorem{remark}[theorem]{Remark}

\begin{document}

\title[Solutions of NPE with Periodic Potential]{ Solutions of  Gross-Pitaevskii Equation with Periodic Potential in Dimension Two.}
\author[Yu. Karpeshina, Seonguk  Kim,  R.Shterenberg]{Yulia Karpeshina, Seonguk Kim,  Roman Shterenberg}
%


\address{Department of Mathematics, Campbell Hall, University of Alabama at Birmingham,
1300 University Boulevard, Birmingham, AL 35294.}
\email{karpeshi@uab.edu}%

\address{Department of Mathematics, Julian Science and Math Center,
Depauw University,
Greencastle, IN 46135.}%
\email{seongukkim@depauw.edu}

\address{Department of Mathematics, Campbell Hall, University of Alabama at Birmingham,
1300 University Boulevard, Birmingham, AL 35294.}
\email{shterenb@uab.edu}%

\address{}
\email{}

\thanks{Supported in part by NSF-grants DMS- 1814664 (Y.K. and R.S) }

\date{\today}


\maketitle
\begin{abstract} Quasi-periodic solutions of  a   nonlinear polyharmonic equation for the case $4l>n+1$ in $\R^n$, $n>1$, are studied.  This includes  Gross-Pitaevskii equation in dimension two ($l=1,n=2$). It is proven that there is an extensive "non-resonant"  set ${\mathcal G}\subset \R^n$ such that for every $\vec k\in \mathcal G$ there is a  solution  asymptotically close to a plane wave 
$Ae^{i\langle{ \vec{k}, \vec{x} }\rangle}$ as $|\vec k|\to \infty $, given $A$ is sufficiently small.\end{abstract}
\section{Introduction}
Let us consider a nonlinear polyharmonic equation with a periodic potential $V(\vec{x})$ and quasi-periodic boundary condition:
\begin{equation}\label{main equation, 4l>n+1}
(-\Delta)^{l}u(\vec{x})+V(\vec{x})u(\vec{x})+\sigma |u(\vec{x})|^{2}u(\vec{x})=\lambda u(\vec{x}),~\vec{x}\in [0,2\pi]^{n},
\end{equation}
\begin{equation}\label{main condition, 4l>n+1}
\begin{cases}
~u(x_{1},\cdots,\underbrace{2\pi}_{s-th},\cdots,x_{n})=e^{2\pi it_{s}}u(x_{1},\cdots,\underbrace{0}_{s-th},\cdots,x_{n}),\\
~\frac{\partial}{\partial x_{s}}u(x_{1},\cdots,\underbrace{2\pi}_{s-th},\cdots,x_{n})=e^{2\pi it_{s}}\frac{\partial}{\partial x_{s}}u(x_{1},\cdots,\underbrace{0}_{s-th},\cdots,x_{n}),\\
\vdots\\
~\frac{\partial^{2l-1}}{\partial x_{s}^{2l-1}}u(x_{1},\cdots,\underbrace{2\pi}_{s-th},\cdots,x_{n})=e^{2\pi it_{s}}\frac{\partial^{2l-1}}{\partial x_{s}^{2l-1}}u(x_{1},\cdots,\underbrace{0}_{s-th},\cdots,x_{n}),\\
~~~~s=1,\cdots,n.
\end{cases}
\end{equation}
where $l$ is an integer, $4l>n+1$, $\vec{t}=(t_{1},\cdots,t_{n}) \in K=[0,1]^{n}$, $\sigma$ is a real number  and $V(\vec{x})$ is a trigonometric polynomial and $\int_{Q} V(\vec{x})d\vec{x}=0,$ $Q=[0,2\pi]^{n}$ being the elementary cell of period $2\pi$. More precisely,
\begin{equation}\label{periodic potential, 4l>n+1}
V(\vec{x})=\sum_{q\neq0, |q|\leq R_{0}}v_{q}e^{i\langle{ q, \vec{x} }\rangle},
\end{equation}
$v_{q}$ being Fourier coefficients.

When $l=1$, $n=1,2,3$, equation \eqref{main equation, 4l>n+1} is a famous Gross-Pitaevskii equation for Bose-Einstein condensate, see e.g. \cite{PS08}.
%
%
In physics papers, e.g. \cite{KS02}, \cite{LO03}, \cite{YB13}, \cite{YD03}, a big variety of numerical computations for Gross-Pitaevskii equation is made. 
However, they are restricted to the one dimensional case and there is a lack of theoretical considerations even for the case $n=1$. In this paper we study the case $4l>n+1$ which includes  $l=1,n=2$. 
%
%

The goal of this paper is to construct asymptotic formulas for $u(\vec{x})$ as $\lambda \to \infty $.
We show that there is an extensive "non-resonant" set ${\mathcal G}\subset \R^n$ such that for every $\vec k\in \mathcal G$ there is a quasi-periodic solution of (\ref{main equation, 4l>n+1}) close to a plane wave 
$Ae^{i\langle{ \vec{k}, \vec{x} }\rangle}$ with $\lambda=\lambda (\vec k, A)$ close to $|\vec k|^{2l}+\sigma |A|^2$ as $|\vec k|\to \infty $ (Theorem \ref{main theorem 4l>n+1}). We assume $A\in \C$ and $|A|$ is sufficiently small:
\begin{equation}|\sigma| |A|^2<\lambda ^{\gamma  }, \ \ \gamma <2l-n.
\label{A}
\end{equation}
Note that $\gamma $ is any negative number for the  Gross-Pitaevskii equation $l=1, n=2$.
The quasi-momentum $\vec t$ in \eqref{main equation, 4l>n+1}
is defined by the formula: $\vec k=\vec t+2\pi j$, $j\in \Z^n$.

We show that the non-resonant set $\mathcal G$ has an asymptotically full measure in $\R^n$:
\begin{equation}
\lim _{R\to \infty}\frac{\left| \mathcal G\cap B_R\right|_n}{|B_R|_n}=1, \label{full}
\end{equation}
where $B_R$ is a ball of radius $R$ in $\R^n$ and $|\cdot |_n$ is Lebesgue measure in $\R^n$.

Moreover,
we investigate a set $\mathcal{D}_{}(\lambda, A)$ of vectors $\vec k\in \mathcal G$, corresponding to a fixed sufficiently large $\lambda $ and a fixed $A$. The set $\mathcal{D}_{}(\lambda, A)$,
defined as a level (isoenergetic) set for $\lambda _{}(\vec
k, A)$, 
\begin{equation} 
{\mathcal D} _{}(\lambda,A)=\left\{ \vec k \in \mathcal{G} _{} :\lambda _{}(\vec k, A)=\lambda \right\},\label{isoset} 
\end{equation}
is proven to be a slightly distorted $n$-dimensional sphere with  
a finite number of holes (Theorem \ref{iso}). For any sufficiently large $\lambda $, it can be described by the formula:
\begin{equation} {\mathcal D}_{}(\lambda, A)=\{\vec k:\vec
k=\varkappa _{}(\lambda, A,\vec{\nu})\vec{\nu},
\ \vec{\nu} \in {\mathcal B}_{}(\lambda)\}, \label{D}
\end{equation}
where ${\mathcal B}_{}(\lambda )$ is a subset of the unit
sphere $S_{n-1}$. The set ${\mathcal B}_{}(\lambda )$ can be
interpreted as a set of possible directions of propagation for 
almost plane waves. Set ${\mathcal B}_{
}(\lambda )$ has an asymptotically full
measure on $S_{n-1}$ as $\lambda \to \infty $:
\begin{equation}
\left|{\mathcal B}_{}(\lambda )\right|=_{\lambda \to \infty
}\omega _{n-1} +O\left(\lambda^{-\delta }\right), \ \ \delta >0,\label{B}
\end{equation}
here $|\cdot |$ is the standard surface measure on $S_{n-1}$, $\omega _{n-1} =|S_{n-1}|$.
The value $\varkappa _{}(\lambda ,A,\vec \nu )$ in (\ref{D}) is
the ``radius" of ${\mathcal D}_{}(\lambda,A)$ in a direction
$\vec \nu $. The function $\varkappa _{}(\lambda ,A,\vec \nu
)-(\lambda-\sigma |A|^2)^{1/2l}$ describes the deviation of ${\mathcal
D}_{}(\lambda,A)$ from the perfect sphere (circle) of the radius
$(\lambda-\sigma |A|^2)^{1/2l}$ in $\R^n$. It is proven that the deviation is asymptotically
small:
\begin{equation} \varkappa _{}(\lambda ,A, \vec \nu
)=_{\lambda \to \infty} \left(\lambda-\sigma |A|^2\right)^{1/2l}+O\left(\lambda ^{-\gamma _1}\right),\ \ \gamma _1>0.
\label{h}
\end{equation}

To prove the results above, we consider the term $V+\sigma|u|^{2}$ in equation (\ref{main equation, 4l>n+1}) as a periodic potential and formally change the nonlinear equation to a linear equation with an unknown potential $V(\vec{x})+\sigma |u(\vec{x})|^{2}$:
\begin{equation*}
(-\Delta)^{l}u(\vec{x})+\big(V(\vec{x})+\sigma |u(\vec{x})|^{2}\big)u(\vec{x})=\lambda u(\vec{x}).
\end{equation*} 
Further, we use known results  for linear polyharmonic equations with periodic potentials.
To start with, we consider a linear operator in $L^{2}(Q)$ 
described by the formula:
\begin{equation}H(\vec t)=(-\Delta)^{l}+V,\label{linear oper 4l>n+1}\end{equation}
and quasi-periodic boundary condition (\ref{main condition, 4l>n+1}).
The free operator $H_{0}(\vec t)$, corresponding to $V=0$, has eigenfunctions given by:
\begin{equation}\psi_{j}(\vec{x})=e^{i\langle{ \vec{p}_{j}(\vec t), \vec{x} }\rangle},~~\vec{p}_{j}(\vec t):=\vec t+2\pi j,~j \in \mathbb{Z}^{n},~\vec t \in K,\label{0}\end{equation}
and the corresponding eigenvalues $p_{j}^{2l}(\vec t):=|\vec{p}_{j}(\vec t)|^{2l}$. 
Perturbation theory for a linear operator $H(\vec t)$ with a periodic potential $V$ is developed in 
\cite{K97}. It is shown that at high energies, there is an extensive set of generalized eigenfunctions being close to plane waves. Below (See Theorem \ref{thm1 4l>n+1}), we describe this result in details.
Now, we define a map ${\mathcal M}: L^{\infty}(Q) \rightarrow L^{\infty}(Q)$ by the formula:
\begin{align}\label{def of A 4l>n+1}
{\mathcal M}W(\vec{x})=V(\vec{x})+\sigma|u_{\tilde{W}}(\vec{x})|^{2}.
\end{align}
Here, $\tilde{W}$ is a shift of $W$ by a constant such that $\int_{Q}\tilde{W}(\vec{x})d\vec{x}=0$,
\begin{align}\label{tilde W 4l>n+1}
\tilde{W}(\vec{x})=W(\vec{x})-\frac{1}{(2\pi)^{n}}\int_{Q}W(\vec{x})d\vec{x},
\end{align}
and $u_{\tilde{W}}$ is an eigenfunction of the linear operator $(-\Delta)^{l}+\tilde{W}$ 
with the boundary condition (\ref{main condition, 4l>n+1}). 
Next, we consider a sequence $\{W_{m}\}_{m=0}^{\infty}$:
\begin{align}\label{def of successive sequence A 4l>n+1}
W_{0}=V+\sigma |A|^2,\ \ \ {\mathcal M}W_{m}=W_{m+1}.
\end{align}
Note that the sequence is well-defined by induction, since  for each $m=1,2,3,\dots$ and $\vec t$ in a non-resonant set ${\mathcal G}$ described in Section 2, there is an eigenfunction $u_{{m}}(\vec{x})$ corresponding to the potential $\tilde{W}_{m}$:
%
$$H_{{m}}(\vec t)u_{{m}} =\lambda_{{m}}u_{{m}},$$
$$H_{{m}}(\vec t)u_{{m}} :=(-\Delta)^{l}u_{{m}}+\tilde{W}_{m}u_{{m}},$$
 $\lambda_{{m}}$, $u_{m}$ being defined by formal series of the form \eqref{series pj 4l>n+1}--\eqref{series E formula 4l>n+1}, \eqref{def of u 2l>n} with $\tilde W_m$ instead of $V$. Those series are proven to be convergent, thus justifying our construction. 
Next, we prove that the sequence $\{W_{m}\}_{m=0}^{\infty}$ is a Cauchy sequence of periodic functions in $Q$ with respect to a norm 
\begin{align}\label{def of star norm 4l>n+1}
\|W\|_{*}=\sum_{q\in \mathbb{Z}^{n}}|w_{q}|,
\end{align}
$w_{q}$ being Fourier coefficients of $W$.
%
This implies that 
$$W_{m}\rightarrow W~\mbox{with respect to the norm}~\|\cdot\|_{*},~W\mbox{ is a periodic function}.$$
Further, we show that 
$$u_{{m}}\rightarrow u_{\tilde W} ~\mbox{in}~L^{\infty}(Q),\ \ \  \ \  \lambda_{{m}}\rightarrow \lambda_{\tilde W} ~\mbox{in}~\mathbb{R},$$
where $u_{\tilde W}$, $ \lambda_{\tilde W}$ correspond to the potential ${\tilde W}$ (via \eqref{series pj 4l>n+1}--\eqref{series E formula 4l>n+1}, \eqref{def of u 2l>n} with $\tilde W$ instead of $V$).
%
It follows from (\ref{def of A 4l>n+1}) and (\ref{def of successive sequence A 4l>n+1}) that ${\mathcal M}W=W$ and, hence, $u_{}:=u_{\tilde W}$ solves the nonlinear equation 
 (\ref{main equation, 4l>n+1}) with quasi-periodic boundary condition (\ref{main condition, 4l>n+1}). 
%

Note, that the results of the paper can be easily generalized for the case of  a sufficiently smooth potential $V(x)$. Generalization for the case $l=1, n=3$ (Gross-Pitaevskii equation in dimension three) is also possible. However, it requires more subtle considerations than here and will be done in a forthcoming paper.

The paper is organized  as follows. In Section 2, we introduce  results for the linear operator $(-\Delta )^l+V$ which include the perturbation formulas for an eigenvalue and its spectral projection. In Section 3, we prove  existence of  solutions of the equation \eqref{main equation, 4l>n+1} with boundary condition \eqref{main condition, 4l>n+1}
and investigate their properties. Isoenergetic surfaces are also introduced and described there.
%
%
%
%
%
%
%
%
%
\section{Linear Operator}

Let us consider an operator
\begin{equation}\label{linear eq 4l>n+1 in R}
H=(-\Delta)^{l}+V,
\end{equation} 
in $L^{2}(\mathbb{R}^{n})$, $4l>n+1$ and $n \geq 2$ where $l$ is an integer and $V(\vec{x})$ is defined by \eqref{periodic potential, 4l>n+1}. Since  $V(\vec x)$ is periodic with an elementary cell $Q$,  the spectral study of (\ref{linear eq 4l>n+1 in R}) can be reduced to that of a family of Bloch operators $H(\vec t)~\mbox{in}~L^{2}(Q),~\vec t\in K$ (see formula \eqref{linear oper 4l>n+1} and quasi-periodic conditions (\ref{main condition, 4l>n+1})).

The free operator $H_{0}(\vec t)$, corresponding to $V=0$, has eigenfunctions given by \eqref{0}
%
%
and the corresponding eigenvalue is $p_{j}^{2l}(\vec t):=|\vec{p}_{j}(\vec t)|^{2l}$. Next, we describe an isoenergetic surface of $H_0(\vec t)$ in $K$.
To start with, we consider the sphere $S(k)$ of radius $k$ centered at the origin in $\R^n$.  For each $j\in \Z^n$ such that $j+K\cap S(k)\neq \emptyset $, $K:=[0,1]^{n}$, we translate the corresponding piece of $S(k)$ into $K$, thus obtaining the sphere of radius $k$ ``packed" into $K$.
We denote it by $S_{0}(k)$. Namely, 
$$S_{0}(k)=\big\{\vec t \in K:~ \mbox{there is a }~j\in \mathbb{Z}^{n}~\mbox{such that}~p_{j}^{2l}(\vec t)=k^{2l}\big\}.$$  
Obviously,  operator $H_{0}(\vec t)$ has an eigenvalue equal to $k^{2l}$ if and only if $\vec t \in S_{0}(k)$. For this reason, $S_{0}(k)$ is called an isoenergetic surface of $H_{0}(\vec t)$.
When $\vec t$ is a point of self-intersection of $S_{0}(k)$, there exists $q\neq j$ such that 
\begin{align}\label{pq=pj 2l>n}
p_{q}^{2l}(\vec t)=p_{j}^{2l}(\vec t).
\end{align}
In other words,  there is a non-simple eigenvalue of $H_{0}(\vec t)$. We remove from  
the set $S_{0}(k)$  the $(k^{-n+1-\delta })$-neighborhoods of all self-intersections \eqref{pq=pj 2l>n}. We call the remaining set  a non-resonant set and denote is by  $\chi_{0}(k,\delta)$, The removed neighborhood of self-intersections is relatively small  and, therefore,   $\chi_{0}(k,\delta)$ has asymptotically full measure with respect to $S_{0}(k)$:
$$\frac{\left|\chi_{0}(k,\delta)\right|}{\left|S_{0}(k)\right|}=1+O(k^{-\delta/8}),$$
here and below $|\cdot |$ is Lebesgue measure of a surface in $\R^n$.
It can be easily shown that for any $\vec t\in \chi_{0}(k,\delta)$,  there is a unique $j\in \Z^n$ such that $p_j^{2l}(\vec t)=k^{2l}$ and
\begin{equation}\min _{q\neq j}\left|  p_q^{2l}(\vec t)-p_j^{2l}(\vec t)\right| >k^{2l-n-\delta }.\label{in} \end{equation}
This means that the distance from $p_{j}^{2l}(\vec t)$ to the nearest eigenvalue $p_{q}^{2l}(\vec t)$, $q\neq j$ is greater than $k^{2l-n-\delta }$. If $2l>n$, then this distance  is large and standard perturbation series can be constructed for $p_{j}^{2l}(\vec t)$, $t \in \chi_{0}(k,\delta)$.
 However, the denseness of the eigenvalues increases infinitely when $k\to \infty $ and $2l<n$. Hence, eigenvalues of the free operator $H_{0}(\vec t)$  strongly interact with each other  when $2l<n$, the case $2l=n$ being intermediate. Nevertheless,  the perturbation series for eigenvalues and their spectral projections  were constructed in \cite{K97} for $4l>n+1$ when $\vec t$ belongs to a non-resonant set $\chi_{1}$.
%


%
%
%
\begin{lemma}\label{lemma t 4l>n+1} For any $0<\beta<1$,   $0<2\delta<(n-1)(1-\beta)$ and  sufficiently large  $k>k_{0}(\beta, \delta)$,  
there is a non-resonant set $\chi_{1}(k,\beta,\delta)$ such that for any $t\in \chi_{1}(k,\beta,\delta)$ there is a unique $j\in \Z^n$: $p_j^{2l}(\vec t)=k^{2l}$ and
 if $\vec t$ is in the $(k^{-n+1-2\delta })$-neighborhood of  $\chi_{1}(k,\beta,\delta)$ in $K$,
then for $z \in C_{0}=\{z \in \mathbb{C}: |z-k^{2l}|=k^{2l-n-\delta}\}$ we have 
\begin{equation}\min _{i\in \Z^n}\left|p_i^{2l}(\vec t)-z\right| >k^{2l-n-\delta }, \label{in-a} \end{equation}
\begin{equation} 200|p_{i}^{2l}(\vec t)-z||p_{i+q}^{2l}(\vec t)-z|>k^{2\gamma_{2}},~~i \in \mathbb{Z}^{n},~|q| <k^{\beta},~q\neq0, \label{main}
\end{equation}
here and below:
\begin{align}\label{gamma 2 4l bigger n+1}
2\gamma_{2}=4l-n-1-\beta(n-1)-2\delta >0.
\end{align}
The non-resonant set $\chi_{1}(k,\beta,\delta)$ has an asymptotically full measure on $S_{0}(k)$:
\begin{equation*} 
\frac{s\Big(S_{0}(k)\setminus\chi _1(k,\beta,\delta )\Big)}{s\Big(S_{0}(k)\Big)}=O(k^{-\delta/8}).
\end{equation*}
\end{lemma}
%
%
%
\begin{theorem}\label{thm1 4l>n+1} 
Under the conditions of Lemma \ref{lemma t 4l>n+1}, 
there exists a single eigenvalue of the operator $H(\vec t)$ 
in
the interval 
$\varepsilon (k,\delta )\equiv (k^{2l}-k^{2l-n-\delta },
k^{2l}+k^{2l-n-\delta })$. 
It is given by the series:
\begin{equation}\label{series pj 4l>n+1}
\lambda (\vec t)=p_j^{2l}(\vec t)+\sum _{r=2}^{\infty }g_r(k,t), 
\end{equation}
converging absolutely, where the index $j$
is uniquely determined from the relation 
$p_j^{2l}(\vec t)\in \varepsilon (k,\delta )$ and
\begin{equation}\label{series pj formula 4l>n+1}
g_{r}(k,\vec t)=\frac{(-1)^{r}}{2\pi ir}\Tr \oint_{C_{0}} \left((H_{0}(\vec t)-z)^{-1}V\right)^{r}dz.
\end{equation}
The spectral projection, corresponding to $\lambda (\vec t)$ is given by 
the series:
\begin{equation}\label{series E 4l>n+1}
E( t)=E_j+\sum _{r=1}^{\infty }G_r(k,t), 
\end{equation}
which converges in the trace class $S_1$ uniformly,
where
\begin{align}\label{series E formula 4l>n+1}
G_{r}(k,\vec t)=\frac{(-1)^{r+1}}{2\pi i}\oint_{C_{0}}\left((H_{0}(\vec t)-z)^{-1}V\right)^{r}(H_{0}(\vec t)-z)^{-1}dz.
\end{align}
Moreover, for coefficients $g_{r}(k,\vec t), G_{r}(k,\vec t)$, the following estimates hold:
\begin{align}\label{ii-a}
&|g_{r}(k,\vec t)|<k^{2l-n-\delta}k^{-\gamma_{2}r},\end{align}
\begin{align}\label{ii}
\|G_{r}(k,\vec t)\|_{S_1}\leq \hat{v}k^{-\gamma_{2}r}, ~~\hat{v}=cR_{0}^{n}\max\limits_{m\in \mathbb{Z}^{n}}|v_{m}|.
\end{align}
\end{theorem}

\begin{remark} \label{r} We use the following norm $\|T\|_1$ of an operator $T$ in $l_2(\Z^2)$:
$$\|T\|_1=\max _{i}\sum _p|T_{pi}|.$$ It can be easily seen from construction in \cite{K97} that estimates  \eqref{ii} hold with respect to this norm too. 
\end{remark}
Let us introduce the notations:
\begin{equation}
T(m)\equiv \frac{\partial ^{\mid m\mid }}{\partial t_1^{m_1}
\partial t_2^{m_2}...\partial t_n^{m_n}}, \label{1.1.20a}
\end{equation}
$$\mid m\mid \equiv m_1+m_2+...+m_n,\ m!\equiv m_1!m_2!...m_n!,$$
$$0\leq \mid m\mid <\infty,\ T(0)f\equiv f.$$
The following theorem and corollary are proven in \cite{K97}.
\begin{theorem} \label{2.3} Under the conditions of Theorem 2.2, the series (\ref{series pj 4l>n+1}), (\ref{series E 4l>n+1})
can be differentiated with respect to $\vec t$ any number of times, and 
they retain their asymptotic character. Coefficients $g_r(k,\vec t)$ and 
$G_r(k,\vec t)$ satisfy the following estimates in the $(k^{-n+1-2\delta })$-neighborhood in $\C^n$ of the nonsingular set $\chi _1(k,\beta,\delta )$:
\begin{equation}
\mid T(m)g_r(k,\vec t)\mid <m!k^{2l-n-\delta}(\hat{v}k^{-\gamma_{2}})^{r}k^{\mid m\mid (n-1+2\delta )} ,\label{2.2.32a}
\end{equation}
\begin{equation}
\| T(m)G_r(k,\vec t)\|_{1}<m!(\hat{v}k^{-\gamma_{2}})^{r}k^{\mid m\mid (n-1+2\delta )}. \label{2.2.33a}
\end{equation}
\end{theorem}
\begin{corollary} \label{derivatives} There are the estimates for the perturbed eigenvalue and its
spectral projection:
\begin{equation}
\mid T(m)(\lambda (\vec t)-p_j^{2l}(\vec t))\mid <cm!k^{(n-1+2\delta)|m|}k^{2l-n-\delta-2\gamma_2}, \label{2.2.34b}
\end{equation}
\begin{equation} 
\| T(m)(E(\vec t)-E_j)\|_{1} <cm!k^{(n-1+2\delta)|m|}k^{-\gamma_2}. \label{2.2.35}
\end{equation}
\end{corollary}

\begin{corollary} \label{def1}There is a one-dimensional space of Bloch eigenfunctions $u_0$ corresponding to the projection $E(t)$ given by  \eqref{series E 4l>n+1}.
They are given by the formula:
\begin{align}\label{def of u 2l>n}
 u_0(\vec{x})&=AE(\vec t)e^{i\langle{ \vec{p}_{j}(\vec t), \vec{x} }\rangle}=A\sum_{m \in \mathbb{Z}^{n}}E(\vec t)_{mj}e^{i\langle{ \vec{p}_{m}(\vec t), \vec{x} }\rangle}\\
\nonumber &=Ae^{i\langle{ \vec{p}_{j}(\vec t), \vec{x} }\rangle}\Bigg(1+\sum_{q\neq 0}\frac{v_{q}}{p_{j}^{2l}(\vec t)-p_{j+q}^{2l}(\vec t)}e^{i\langle{\vec{p}_{q}(\vec{0}), \vec{x} }\rangle}+\cdot\cdot\cdot\Bigg),~~j,q \in \mathbb{Z}^{n},\ A\in \C.
\end{align}
\end{corollary}
Let $\tilde \chi _1(k,\beta,\delta)\subset S(k)$ be the image of $\chi _1(k,\beta,\delta)\subset S_0(k)$ on the sphere $S(k)$:
\begin{equation}
\tilde \chi _1(k,\beta,\delta)=\{\vec p_j(\vec t)\in S(k):\  \vec t \in  \chi _1(k,\beta,\delta)\}.\end{equation}
Note that $\tilde \chi _1(k,\beta,\delta)$ is well-defined, since $ \chi _1(k,\beta,\delta)$ does not contain self intersections of $S_0(k)$.
Let $\mathcal B(\lambda )\subset S_{n-1}$ be the set of directions corresponding to the nonsingular set $\tilde \chi _1(k,\beta,\delta)$: \begin{equation} \label{formulaB}
\mathcal B(\lambda)=\big\{\vec \nu \in S_{n-1}: k\vec \nu \in \tilde \chi _1(k,\beta,\delta)\big\}, \ k^{2l}=\lambda. \end{equation}
The set $\mathcal B(\lambda)$ can be interpreted as a set of possible directions of propagation for almost plane waves \eqref{def of u 2l>n}.
We define the non-resonance set $\mathcal G\subset \R^n$ as the union of all $ \tilde \chi _1(k,\beta,\delta)$:
\begin{equation} \label{G}
\mathcal G=\bigcup\limits _{k>k_0(\beta, \delta )}\tilde  \chi _1(k,\beta,\delta)
\end{equation}
Further we denote vectors of $\mathcal G$ by $\vec k$. Formulas \eqref{formulaB}, \eqref{G} yield:
\begin{equation} \label{G+}
\mathcal G=\big\{\vec k=k\vec \nu: \vec \nu \in \mathcal B(k^{2l}), \ k>k_0(\beta,\delta )\big\}. \end{equation}
Since any vector $\vec k$  can be written as $\vec k=\vec p_j(t)$ in a unique way, formula \eqref{G} yields:
\begin{equation} \label{G*}
\mathcal G=\big\{\vec p_j(\vec t):\ \vec t \in  \chi _1(k,\beta,\delta),\ \mbox{where } k=p_j(\vec t),\ k>k_0(\beta,\delta )\big\}. \end{equation}
Let $\lambda (\vec k)$ be defined by \eqref{series pj 4l>n+1}, where $\vec k=\vec p_j(\vec t)$.

Next, we describe isoenergetic surfaces for  the operator \eqref{linear eq 4l>n+1 in R}. The set $\mathcal{D}_{}(\lambda)$,
defined as a level (isoenergetic) set for $\lambda _{}(\vec
k)$, \begin{equation} {\mathcal D} _{}(\lambda )=\left\{ \vec k \in \mathcal{G} _{} :\lambda _{}(\vec k)=\lambda \right\}.\label{isoset-lin} \end{equation}
\begin{lemma}\label{L:2.12a} For any sufficiently large $\lambda $, $\lambda >k_0(\beta,\delta )^{2l}$, and for every
$\vec{\nu}\in\mathcal{B} (\lambda)$, there is a
unique $\varkappa =\varkappa (\lambda ,\vec{\nu})$ in the
interval $$
I:=[k-k^{-n+1-2\delta},k+k^{-n+1-2\delta}],\quad
k^{2l}=\lambda , $$ such that \begin{equation}\label{2.70}
\lambda (\varkappa \vec{\nu})=\lambda . \end{equation}
Furthermore, $|\varkappa - k| \leq ck^{2l-n-\delta-2\gamma_2-2l+1}=ck^{-\gamma_1}$, $\gamma_1=4l-2-\beta(n-1)-\delta>0$.
\end{lemma} 
The Lemma easily follows from \eqref{2.2.34b} for $|m|=1$. 


\begin{lemma} \label{L:2.13a} \begin{enumerate} \item For any sufficiently
large $\lambda $, $\lambda >k_0(\beta,\delta )^{2l}$, the set $\mathcal{D}(\lambda )$, defined by \eqref{isoset-lin} is a distorted
sphere with holes; it is described by the formula:
\begin{equation} 
{\mathcal D}_{}(\lambda)=\{\vec k:\vec
k=\varkappa _{}(\lambda,\vec{\nu})\vec{\nu},
\ \vec{\nu} \in {\mathcal B}_{}(\lambda)\},\label{May20} 
\end{equation} where $\varkappa (\lambda,\vec \nu)=k+h (\lambda, \vec \nu)$ and $h (\lambda, \vec \nu)$ obeys the
inequalities:
\begin{equation}\label{2.75a}
|h|<ck^{-\gamma_1},\quad
\left|\nabla _{\vec \nu }h\right| <
ck^{-\gamma_1+n-1+2\delta}=ck^{-2\gamma_2+\delta}.
\end{equation}

\item The measure of $\mathcal{B}(\lambda)\subset S_{n-1}$ satisfies the
estimate \eqref{B}.


\item The surface $\mathcal{D}(\lambda)$ has the measure that is
asymptotically close to that of the whole sphere of the radius $k$ in the sense that
\begin{equation}\label{2.77}
\bigl |\mathcal{D}(\lambda)\bigr|\underset{\lambda \rightarrow
\infty}{=}\omega _{n-1}k^{n-1}\bigl(1+O(k^{-\delta})\bigr),\quad \lambda =k^{2l}.
\end{equation}
\end{enumerate}
\end{lemma}
The proof is based on Implicit Function Theorem.
%
%
%
%
%
%
%
%
%
%
%
\section{Proof of The Main Result}
First, we prove that  $\{W_{m}\}_{m=0}^{\infty}$ in \eqref{def of successive sequence A 4l>n+1} is a Cauchy sequence with respect to  the norm defined by (\ref{def of star norm 4l>n+1}).
%
%
%
Further we  need the following  obvious properties of norm  $\|\cdot\|_{*}$:
\begin{align}\label{properties 1 2l>n}
\|f\|_{*}=\|\bar{f}\|_{*},\ \
\|\Re(f)\|_{*}\leq\|f\|_{*},\ \
\|\Im(f)\|_{*}\leq\|f\|_{*},\ \
\|fg\|_{*}\leq \|f\|_{*}\|g\|_{*}.
\end{align}
where $\Re(f)$ and $\Im(f)$ are real and imaginary part for $f$, respectively.
%

%
%
 We define  the value $k_{1}=k_{1}(\|V\|_{*},\delta,\beta)$ as
\begin{equation}
k_{1}(\|V\|_{*},\delta,\beta)=\max\Big\{
\left(16\|V\|_{*}\right)^{1/\gamma _2},\ k_0(\beta,\delta )\Big\},
\label{k_{1}}
\end{equation}
$\gamma_{2}>0$ being defined by \eqref{gamma 2 4l bigger n+1} and
 $k_0(\beta, \delta )$ being as in Corollary \ref{lemma t 4l>n+1}.
\begin{lemma}\label{main lemma  1 2l>n} 
The following inequalities hold for any  $m=1,2,\dots$:
\begin{align}
\|\tilde{W}_{m}-V\|_{*}\leq 8|\sigma||A|^{2}\|V\|_{*}k^{-\gamma_{2}}, \label{mmm}
\end{align}
\begin{align}
\|W_{m}-W_{m-1}\|_{*}\leq 4|\sigma||A|^{2}\|V\|_{*}k^{-\gamma_{2}}(|\sigma ||A|^2k^{-\gamma_{0}})^{m-1}, \label{mm}
\end{align}
\begin{align}\label{estimate of difference of E-a}
 \|E_{{m}}(\vec t)-E_{{m-1}}(\vec t)\|_1 
\leq &8|\sigma||A|^{2}\|V\|_{*}k^{-(2l-n-\delta)-\gamma_{2}}(|\sigma ||A|^2k^{-\gamma_{0}})^{m-1},
\end{align}
where $\gamma_{0}=2l-n-2\delta$, $\delta >0$, and $|\sigma ||A|^2<k^{\gamma _0-\delta }$,  $k$ being sufficiently large $k>k_{1}(\|V\|_{*},\delta,\beta)$.
\end{lemma}
\begin{corollary}\label{cauchy sequence 2l>n} 
There is a periodic function $W $   such that $W_{m}$ converges to $W$ with respect to the norm $\|\cdot\|_{*}$:
\begin{align}
\|W-W_{m}\|_{*}\leq 8|\sigma||A|^{2}\|V\|_{*}k^{-\gamma_{2}}(|\sigma ||A|^2k^{-\gamma_{0}})^{m}. \label{mm+}
\end{align}
\end{corollary}
\begin{proof}[Proof of Lemma 3.1]
Let us consider the function \eqref{def of u 2l>n}  written in the form
\begin{align}\label{expression u 2l>n 2}
u_0(\vec{x})=\psi _0(\vec{x})e^{i\langle{ \vec{p}_{j}(\vec t), \vec{x} }\rangle},
\end{align}
where 
\begin{align}\label{change cordinator 2l>n*}
\psi _0(\vec{x})=A\sum_{q\in \mathbb{Z}^{n}}E(\vec t)_{j+q,j}e^{i\langle{ \vec{p}_{q}(\vec{0}), \vec{x} }\rangle},
\end{align}
is  called the periodic part of $u_0$. 

First, we prove \eqref{mm} for $m=1$. It follows from (\ref{def of A 4l>n+1}), (\ref{def of successive sequence A 4l>n+1}) and  (\ref{properties 1 2l>n}) that 
\begin{align}\label{change cordinator 2 2l>n}
\nonumber \big\|W_{1}-W_{0}\big\|_{*} =&
|\sigma| \big\||u_{0}|^{2}-|A|^2\big\|_{*}=|\sigma| \big\||\psi_{0}|^{2}-|A|^2\big\|_{*}\\
\leq& \nonumber|\sigma| \big\||\psi_{0}|^{2}-|A|^2+2i\Im(\bar A\psi_{0})\big\|_{*}=\nonumber |\sigma| \big\|(\psi_{0}-A)(\bar{\psi}_{0}+\bar A)\big\|_{*}\\
\leq& |\sigma|\big\|\psi_{0}-A\big\|_{*}\big\|\bar{\psi}_{0}+\bar A\big\|_{*}.
\end{align}
%
Let us consider
\begin{equation}B_0(z)=(H_0(\vec t)-z)^{-\frac{1}{2}}V(H_0(\vec t)-z)^{-\frac{1}{2}}.\label{M17b*} \end{equation}
Then it follows from Lemma \ref{lemma t 4l>n+1} that:
\begin{align}\label{l}
\max_{z\in C_{0}}\Big\|(H_{0}(\vec t)-z)^{-1}\Big\|_{1}<k^{-2l+n+\delta}.
\end{align}
\begin{equation} \label{M17c*}
\max_{z\in C_{0}}\|B_0(z)\|_1<\|V\|_* k^{-\gamma_2},
\end{equation}
$\gamma _2$ being defined by \eqref{gamma 2 4l bigger n+1}. 
By (\ref{series E formula 4l>n+1}) and (\ref{M17b*}),
\begin{align}\label{proj op c 2l>n*}
G_{r}(k,\vec t)=\frac{(-1)^{r+1}}{2\pi i}\oint_{C_{0}}(H_0(\vec t)-z)^{-\frac{1}{2}}B_0(z)^{r}(H_0(\vec t)-z)^{-\frac{1}{2}}dz.
\end{align}
It is easy to see that
\begin{equation}
\|G_{r}(
k,\vec t)\|_1<\|V \|_*^rk^{-\gamma_{2}r}. \label{M17e*}
\end{equation}
%
%
%
%
%
%
%
Next, by \eqref{change cordinator 2l>n*}, \eqref{series E formula 4l>n+1} and \eqref{ii},
\begin{align}\label{est psiv 2l>n*}
\nonumber \|\psi_{0}-A\|_{*}
 \leq &\Big|AE(\vec t)_{jj}-A\Big|+|A|\sum_{q\in\mathbb{Z}^{n}\setminus\{0\}}\Big|E_{}(\vec t)_{j+q,j}\Big| 
 \\
\leq& |A| \sum_{r=1}^{\infty} \|G_{r}(k,\vec t)\|_{1}
 \leq \|V\|_{*}|A|k^{-\gamma_{2}}(1+o(1)).
\end{align}
%
It follows:
\begin{align}\label{est bar psiv 2l>n*}
\|{\psi}_{0}\|_{*}=\|\bar{\psi}_{0}\|_{*}\leq |A|+ O\big(|A|k^{-\gamma_{2}}\big).
\end{align}
Using (\ref{change cordinator 2 2l>n}),  (\ref{est psiv 2l>n*}) and (\ref{est bar psiv 2l>n*}), we get 
\begin{align*}
&\|W_{1}-W_{0}\|_{*}
\leq 4|\sigma||A|^2\|V\|_{*}k^{-\gamma_{2}}.
\end{align*}
Since $\|\tilde{W}_{1}-V\|_{*}=\|\tilde{W}_{1}-\tilde{W}_{0}\|_{*}\leq \|W_{1}-W_{0}\|_{*}$, we have:
\begin{align}\label{54.5}
&\|\tilde{W}_{1}-V\|_{*} \leq 4|\sigma||A|^2\|V\|_{*}k^{-\gamma_{2}}.
\end{align}
Now, we use  mathematical induction to show simultaneously,
\begin{equation}
\|\tilde{W}_{m}-V\|_{*} \leq 8|\sigma||A|^{2}\|V\|_{*}k^{-\gamma_{2}}, \label{M19b11}
\end{equation}
\begin{equation}
\|W_{m}-W_{m-1}\|_{*} \leq 4|\sigma||A|^{2}\|V\|_{*}k^{-\gamma_{2}}(|\sigma ||A|^2k^{-\gamma_{0}} )^{m-1}. \label{M19b 111}
\end{equation}
Suppose that for all $1\leq s\leq m-1$,
\begin{equation}
\|\tilde{W}_{s}-V\|_{*} \leq 8|\sigma||A|^{2}\|V\|_{*}k^{-\gamma_{2}}, \label{M19b a}
\end{equation}
\begin{equation}
\|W_{s}-W_{s-1}\|_{*} \leq 4|\sigma||A|^{2}\|V\|_{*}k^{-\gamma_{2}}(|\sigma ||A|^2k^{-\gamma_{0}} )^{s-1}. \label{M19b}
\end{equation}
%
%
%
Let, by analogy with \eqref{def of u 2l>n},
\begin{align}\label{expression u 2l>n 1}
u_{{s}}(\vec{x}):=A\sum\limits_{m \in \mathbb{Z}^{n}}E_{{s}}(\vec t)_{m,j}e^{i\langle{ \vec{p}_{m}(\vec t), \vec{x} }\rangle}, 
\end{align}
where $E_{{s}}(\vec t)$ is the spectral projection \eqref{series E 4l>n+1} with the potential $\tilde{W}_{s}$.
Obviously,
\begin{align}\label{expression u 2l>n 2}
u_{{s}}(\vec{x})=\psi_{{s}}(\vec{x})e^{i\langle{ \vec{p}_{j}(\vec t), \vec{x} }\rangle},
\end{align}
where the function, 
\begin{align}\label{change cordinator 2l>n}
\psi_{{s}}(\vec{x})=A\sum_{q\in \mathbb{Z}^{n}}E_{{s}}(\vec t)_{j+q,j}e^{i\langle{ \vec{p}_{q}(\vec{0}), \vec{x} }\rangle},
\end{align}
is   the periodic part of $u_{{s}}$. Clearly,
\begin{equation}\|\psi_{{s}}\|_*\leq  |A|\|E_{{s}}(\vec t)\|_1.\label{M19a}
\end{equation}
Let
\begin{equation}B_s(z)=(H_0(\vec t)-z)^{-\frac{1}{2}}\tilde W_s (H_0(\vec t)-z)^{-\frac{1}{2}}.\label{M17b} \end{equation}
Using     \eqref{M19b a},  \eqref{M17c*} and \eqref{in-a}, we easily obtain:
\begin{equation} \label{M17c}
\|B_s(z)\|_1\leq 8|\sigma||A|^{2} \|V\|_{*}k^{-2l+n+\delta-\gamma_2}+\|V\|_*k^{-\gamma_2}\leq 2\|V\|_*k^{-\gamma_2},\ \  z \in C_0,
\end{equation}
for any  $1\leq s\leq m-1$.
It is easy to see now that
\begin{equation}
\|G_{s,r}(k,\vec t)\|_1\leq \big(4\|V\|_*k^{-\gamma_2}\big)^{r},\ \ 1\leq s\leq m-1, \label{M18a}
\end{equation}
here  $G_{s,r}(k,\vec t)$ is given by  \eqref{series E formula 4l>n+1}
with $\tilde W_s$ instead of $V$.
It follows:
\begin{align}\label{estimate of E}
\nonumber \|E_{{s}}(\vec t)\|_1\leq&
1+\sum_{r=1}^{\infty}\|G_{{s},r}(k,\vec t)\|_1 \\
\leq& 1+8\|V\|_*k^{-\gamma_2}\leq2, \ \ 1\leq s\leq m-1.
\end{align}
Next, we note that
\begin{align}\label{estimate of difference of B}
\nonumber& \max_{z\in C_{0}}\|B_{m-1}^r(z)-B^r_{m-2}(z)\|_1\\
\nonumber\leq& \max_{z\in C_{0}}\|B_{m-1}(z)-B_{m-2}(z)\|_1\left(\|B_{m-1}(z)\|_1+\|B_{m-2}(z)\|_1\right)^{r-1}\\
\leq  &k^{-(2l-n-\delta)}\|\tilde{W}_{m-1}-\tilde{W}_{m-2}\|_{*}\Big(4\|V\|_*k^{-\gamma_2}\Big)^{r-1}.
\end{align}
Hence,
\begin{equation}\label{estimate of difference of G}
\|G_{{m-1},r}(k,\vec t)-G_{{m-2},r}(k,\vec t)\|_1\leq k^{-(2l-n-\delta)}\|\tilde{W}_{m-1}-\tilde{W}_{m-2}\|_{*}\Big(4\|V\|_*k^{-\gamma_2}\Big)^{r-1}.
\end{equation}
Estimate  (\ref{estimate of difference of G}) yields:
%
%
%
\begin{align}\label{estimate of difference of E}
 \nonumber \|E_{{m-1}}(\vec t)-E_{{m-2}}(\vec t)\|_1 
\leq &\nonumber\sum_{r=1}^{\infty}
\|G_{{m-1},r}(k,\vec t)-G_{{m-2},r}(k,\vec t)\|_1\\
\leq &2k^{-(2l-n-\delta)}\|\tilde{W}_{m-1}-\tilde{W}_{m-2}\|_{*}.
\end{align}
%
%
%
%
%
Next, considering as in \eqref{change cordinator 2 2l>n},  we obtain:
\begin{align}
\big\|W_{m}-W_{m-1}\big\|_{*}
\leq |\sigma|\big\|\psi_{{m-1}}-\psi_{{m-2}}\big\|_{*}\big\|\bar{\psi}_{{m-1}}+\bar{\psi}_{{m-2}}\big\|_{*},\label{**}
\end{align}
and,
 hence, by \eqref{change cordinator 2l>n},
 \begin{equation} \label{W-m}
 \|W_{m}-W_{m-1}\|_{*}\leq |\sigma ||A|^2\|E_{{m-1}}(\vec t)-E_{{m-2}}(\vec t)\|_1
\left(\|E_{{m-1}}(\vec t)\|_1+
\|E_{{m-2}}(\vec t)\|_1\right).
\end{equation}
%
Using  (\ref{estimate of E}) and (\ref{estimate of difference of E}),  we obtain
\begin{align}
\|W_{m}-W_{m-1}\|_{*}
\leq &8|\sigma||A|^2 k^{-(2l-n-\delta)}\|\tilde{W}_{m-1}-\tilde{W}_{m-2}\|_{*}. \label{***}\end{align}
Considering  $\|\tilde{W}_{m-1}-\tilde{W}_{m-2}\|_{*}\leq \|{W}_{m-1}-{W}_{m-2}\|_{*}$ and using \eqref{M19b} for $s=m-1$, we arrive at the estimate:
\begin{align} \label{main results 123}
 \|W_{m}-W_{m-1}\|_{*}\leq &8|\sigma||A|^2 k^{-(2l-n-\delta)}4|\sigma||A|^{2}\|V\|_{*}k^{-\gamma_{2}}\big(|\sigma||A|^2k^{-\gamma_{0}}\big)^{m-2}\\
\nonumber\leq & 4|\sigma||A|^{2}\|V\|_{*}k^{-\gamma_{2}}\big(|\sigma||A|^2k^{-\gamma_{0}}\big)^{m-1}, 
\end{align}
when $k>k_{1}(\|V\|_{*},\delta,\beta)$.
Further, \eqref{main results 123} and \eqref{54.5} enable the estimate
\begin{align*} 
\nonumber \|\tilde{W}_{m}-V\|_{*}\leq& \|\tilde{W}_{m}-\tilde{W}_{m-1}\|_{*}+\|\tilde{W}_{m-1}-\tilde{W}_{m-2}\|_{*}+\cdots+\|\tilde{W}_{1}-V\|_{*}\\
\leq &8|\sigma||A|^{2}\|V\|_{*}k^{-\gamma_{2}},
\end{align*}
which completes the proof of \eqref{mmm} and \eqref{mm}. Using \eqref{estimate of difference of E}, we obtain \eqref{estimate of difference of E-a}.
\end{proof}
\begin{lemma} \label{lemma3 2l>n}
Suppose $\vec t$ belongs to the $(k^{-n+1-2\delta })$-neighborhood in $K$ 
of the 
non-resonant set $\chi _1(k,\beta,\delta )$. Then for every sufficiently 
large $k>k_{1}(\|V\|_{*},\delta,\beta)$ and every $A\in \C: |\sigma||A|^2 <k^{\gamma_0-\delta }$,  the sequence $E_m(\vec t)$ converges with respect to $\|\cdot \|_1$ to a one-dimensional spectral projection $E_{\tilde W}(\vec t)$ of $H_0(t)+\tilde W$:
\begin{equation}
\|E_m(\vec t)-E_{\tilde W}(\vec t)\|_{1}\leq 
8\|V\|_*k^{-\gamma_2}(|\sigma ||A|^2k^{-\gamma_{0}})^{m+1}. \label{mm++}
\end{equation} 
The projection $E_{\tilde W}(\vec t)$ is given by the series  \eqref{series E 4l>n+1}, \eqref{series E formula 4l>n+1} with $\tilde W$ instead of $V$. The series converges  with respect to $\|\cdot \|_1$:
\begin{align}\label{ii-b}
\|G_{r}(k,\vec t)\|_{1}\leq \left(2\|V\|_*k^{-\gamma_{2}}\right)^{r}\end{align}
\end{lemma}
\begin{proof}
Let  $B(z)$ be given by \eqref{M17b} with $\tilde W$ instead of $\tilde W _s$.  Obviously, $B(z)$ is the limit of $B_m(z)$  in $\|\cdot \|_1$-norm. The estimate \eqref{M17c} yields:
\begin{equation} \label{M17c-1}
\|B(z)\|_1\leq 2\|V\|_*k^{-\gamma_2},\ \  z \in C_0.
\end{equation}
It follows that the perturbation series for the resolvent of $H_0(\vec t)+W$ converges with respect to $\|\cdot \|_1$ norm on $C_0$. Integrating the series of $z$ we obtain that $E(t)$ admits the expansion  \eqref{series E 4l>n+1}, \eqref{series E formula 4l>n+1} and \eqref{ii-b} holds.  Obviously, $G_r$ corresponding to $\tilde W$ is the limit of  $G_{m,r}$ in $\|\cdot\|_1$ norm.
Summing the estimates \eqref{estimate of difference of E}, we obtain \eqref{mm++}.

\end{proof}

%
%
%
\begin{definition} \label{def2} Let $u(\vec x)$ be defined as in Corollary  \ref{def1} for the potential  $\tilde{W}(\vec x)$. Let $\psi (\vec x)$ be the periodic part of $u(\vec x)$.\end{definition}

The next lemma follows from the estimate \eqref{mm++}.

\begin{lemma} \label{lemma3 2l>n'}
Suppose $\vec t$ belongs to the $(k^{-n+1-2\delta })$-neighborhood in $K$ 
of the 
non-resonant set $\chi _1(k,\beta,\delta )$. Then for every sufficiently 
large $k>k_{1}(\|V\|_{*},\delta,\beta)$ and every $A\in \C: |\sigma||A|^2 <k^{\gamma_0-\delta }$, the sequence $\psi _m(\vec x)$ converges to the function $\psi (\vec x)$ with respect to $\|\cdot \|_*$:
\begin{equation}
\|\psi_{{m}}-\psi \|_{*}\leq 8|A|\|V\|_{*}k^{-\gamma_{2}}(|\sigma ||A|^2k^{-\gamma_{0}})^{m+1}. \label{mm+++}
\end{equation}
\end{lemma}

\begin{corollary} \label{3.7} The sequence
$u_{{m}}$ converges to $u_{}$ in $L^{\infty}(Q)$.\end{corollary}
\begin{corollary}\label{AW=W 2l>n} 
$${\mathcal M}W=W.$$
\end{corollary}
\begin{proof}[Proof of Corollary \ref{AW=W 2l>n}] 
Considering as in (\ref{**}), we obtain:
\begin{equation}
\big\|{\mathcal M}W_{m}-{\mathcal M}W_{}\big\|_{*} 
\leq |\sigma|\big\|\psi_{{m}}-\psi \big\|_{*}\big\|\bar{\psi}_{{m}}+\bar{\psi}\big\|_{*},\label{****}
\end{equation}
It immediately follows from Lemma \ref{lemma3 2l>n} that ${\mathcal M}W_{m}\rightarrow {\mathcal M}W $
 with respect to $\|\cdot \|_*$.
Now, by (\ref{def of successive sequence A 4l>n+1}) and \eqref{****}, we have  ${\mathcal M}W=W$.
\end{proof} 
Let $\lambda_{{m}}(\vec t)$,  $\lambda_{\tilde W}(\vec t)$ be the eigenvalues \eqref{series pj 4l>n+1} corresponding to $\tilde W_m$ and $\tilde W$, respectively.
\begin{lemma} \label{lemma lambda 2l>n} Under conditions of Lemma \ref{lemma3 2l>n}
the sequence
 $\lambda_{{m}}(\vec t)$ converges to $\lambda_{\tilde W}(\vec t)$ being given by \eqref{series pj 4l>n+1} and
 \begin{equation}
|g_r(k,\vec t)|<r^{-1}k^{2l-n-\delta} \left(4\|V\|_*k^{-\gamma_{2}}\right)^{r}, \label{gr}
\end{equation}
where $g_r$ is given by  \eqref{series pj formula 4l>n+1} with $\tilde W$ instead of $V$.
\end{lemma}
\begin{proof}
By perturbation theory, $\lambda _{\tilde W}(\vec t)$ is the limit of  $\lambda _m(\vec t)$ as $m\to \infty $. Let us show that the series  \eqref{series pj  4l>n+1} converges.
Let us consider two projections $E_0=E_j$, $E_1=I-E_j$, here $E_j$ is the spectral projection of $H_0$, see \eqref{series E formula 4l>n+1}. Note that
$$\oint _{C_0} \left(E_1B(z)E_1\right)^r dz=0,\ \ r=1,2,...,$$
since the integrand is holomorphic inside $C_0$. Hence,
\begin{align}\label{LAMBDA}
\nonumber& \oint _{C_0} B(z)^r dz =\oint _{C_0} \big(B(z)^r - \left(E_1B(z)E_1\right)^r\big)dz=\\
\nonumber& \sum _{i_1,...,i_{r+1}=0,1, \exists s: i_s=0}\oint _{C_0 }E_{i_1}B(z)E_{i_2}B(z)....E_{i_r}B(z)E_{i_{r+1}}dz. \end{align}

Obviously, $E_{i_1}B(z)_{i_2}B(z)....E_{i_r}B(z)E_{i_{r+1}}$ is in the trace class $S_1$ if at least one index $i_s$,  $1\leq s\leq r+1$ is zero, since $E_0 \in S_1$. Notice that for the adjoint operator $B^*$ we have $B^*(z)=B(\bar z)$. It follows:
$$\|E_{i_1}B(z)E_{i_2}B(z)....E_{i_r}B(z)E_{i_{r+1}}\|_{S_1}\leq \|B\|^r\leq \|B^*\|_1^{r/2}\|B\|_1^{r/2}<\left(2\|V\|_*k^{-\gamma_{2}}\right)^{r}.$$
Now, we easily obtain \eqref{gr}.
\end{proof}
Considering as in the proof of Theorem \ref{2.3}, one can prove an analogous theorem:

\begin{theorem} \label{2.3-a} Under the conditions of Lemma \ref{lemma3 2l>n} the series (\ref{series pj 4l>n+1}), (\ref{series E 4l>n+1}) for the potential $\tilde W$
can be differentiated with respect to $\vec t$ any number of times, and 
they retain their asymptotic character. Coefficients $g_r(k,\vec t)$ and 
$G_r(k,\vec t)$ satisfy the following estimates in the $(k^{-n+1-2\delta })$-neighborhood in $\C^n$ of the nonsingular set $\chi _1(k,\beta,\delta )$:
\begin{equation}
\mid T(m)g_r(k,\vec t)\mid <m!k^{2l-n-\delta}\left(4\|V\|_*k^{-\gamma_{2}}\right)^{r}k^{\mid m\mid (n-1+2\delta )}, \label{2.2.32a-1}
\end{equation}
\begin{equation}
\| T(m)G_r(k,\vec t)\|_{1}<m!(2\|V\|_*k^{-\gamma_{2}})^{r}k^{\mid m\mid (n-1+2\delta )}. \label{2.2.33a-1}
\end{equation}
\end{theorem}
\begin{corollary} \label{derivatives-1} There are the estimates for the perturbed eigenvalue and its
spectral projection:
\begin{equation}
\left| T(m)\big(\lambda_{\tilde W} (\vec t)-p_j^{2l}(\vec t)\big)\right| <C(\|V\|_*)m!k^{(n-1+2\delta)|m|}k^{2l-n-\delta-2\gamma _2}, \label{2.2.34b-1}
\end{equation}
\begin{equation} 
\| T(m)(E_{\tilde W}(\vec t)-E_j)\|_{1} <C(\|V\|_*)m!k^{(n-1+2\delta)|m|}k^{-\gamma _2}.\label{2.2.35-1}
\end{equation}
In particular,
\begin{equation}
\left| \lambda _{\tilde W}(\vec t)-p_j^{2l}(\vec t)\right| <C(\|V\|_*)k^{2l-n-\delta-2\gamma _2}, \label{2.2.34b-3}
\end{equation}
\begin{equation} 
\| E_{\tilde W}(\vec t)-E_j\|_{1} <C(\|V\|_*)k^{-\gamma _2}, \label{2.2.35-3}
\end{equation}
\begin{equation}
\left|\nabla\lambda _{\tilde W}(\vec t)-2l \vec p_j(\vec t)  p_j^{2l-2}(\vec t) \right|<C(\|V\|_*)k^{2l-1-2\gamma _2+\delta } .\label{2.2.34b-2}
\end{equation}
\end{corollary}

We have the following main result for the nonlinear polyharmonic equation with quasi-periodic condition. 

\begin{theorem}  \label{main theorem 4l>n+1}
Suppose $\vec t$ belongs to the $(k^{-n+1-2\delta })$-neighborhood in $K$ 
of the 
non-resonant set $\chi _1(k,\beta,\delta )$,  $k>k_{1}(\|V\|_{*},\delta,\beta)$ and $A\in \C: |\sigma||A|^2 <k^{\gamma_0-\delta }$. Then,
there is  a function $u(\vec{x})$, depending on $\vec t $ as a parameter, and a real value $\lambda(\vec t)$,
satisfying the  equation 
\begin{align}\label{main equation 4l>n+1}
(-\Delta)^{l}u(\vec{x})+V(\vec{x})u(\vec{x})+\sigma |u(\vec{x})|^{2}u(\vec{x})=\lambda u(\vec{x}),~\vec{x}\in 
Q,
\end{align}
and  the quasi-periodic boundary condition (\ref{main condition, 4l>n+1}). The following  formulas hold:
\begin{align} \label{solution construction 2l>n}
u( \vec{x})=&Ae^{i\langle{ \vec{p}_{j}(\vec t), \vec{x} }\rangle}\left(1+\tilde{u}( \vec{x})\right),\\ 
\lambda(\vec{t})=&p_j^{2l}(\vec t)+\sigma|A|^2+O\left( \left(k^{2l-n-\delta}+\sigma |A|^2 \right) k^{-2\gamma _2}\right),~ \label{kkk}
\end{align} 
where $\tilde{u}( \vec{x})$ is periodic and \begin{equation}
\|\tilde{u}\|_{*}\leq k^{-\gamma_{2}},\ \ \gamma_{2}>0~\mbox{ is defined by \eqref{gamma 2 4l bigger n+1}}. \label{June7}
\end{equation}
\end{theorem}
\begin{proof}
Let us consider the function $u$ given by Definition \ref{def2} and  the value $\lambda_{\tilde{W}}(\vec t)$.  
They solve the equation
\begin{align}\label{semi solution 2l>n}
(-\Delta)^{l}u_{}(\vec{x})+\tilde{W}(\vec{x})u_{}(\vec{x})=\lambda_{\tilde{W}}(\vec t) u_{}(\vec{x}),\ \ \  \vec{x}\in Q,
\end{align}
and $u$ satisfies the quasi-boundary condition (\ref{main condition, 4l>n+1}).
By Corollary \ref{AW=W 2l>n}, we have
$$W(\vec{x})={\mathcal M}W(\vec{x})=V(\vec{x})+\sigma|u(\vec{x})|^{2}.$$
Hence,
\begin{align}\label{tilde W 2l>n}
\nonumber\tilde{W}(\vec{x})=&W(\vec{x})-\frac{1}{(2\pi)^{n}}\int_{Q}W(\vec{x})d\vec{x}
=V(\vec{x})+\sigma|u(\vec{x})|^{2}-\sigma\|u\|_{L^{2}(Q)}^{2}.
\end{align}
Substituting the last expression into \eqref{semi solution 2l>n}, we obtain that  $u( \vec{x})$
satisfies \eqref{main equation 4l>n+1} with 
\begin{equation}\label{lambda06}
\lambda(\vec t)=\lambda_{\tilde{W}}(\vec t)+\sigma\|u\|_{L^{2}(Q)}^{2}
=\lambda_{\tilde{W}}(\vec t)+\sigma |A|^2\sum _{q\in Z^n} \big|\left(E_{\tilde W}\right)_{qj}\big|^2=\lambda_{\tilde{W}}(\vec t)+\sigma |A|^2\left(E_{\tilde W}\right)_{jj}.\end{equation}
Note that $(G_1)_{jj}=0$ and, therefore, $\left(E_{\tilde W}\right)_{jj}=1+O(k^{-2\gamma _2})$.
%
%
Further, by the definition of $u( \vec{x})$, we have
\begin{align}\label{u06}
&u( \vec{x}):=Ae^{i\langle{ \vec{p}_{j}(\vec t), \vec{x} \rangle}}\sum\limits_{q\in \Z^n}\left(E_{\tilde W}\right)_{q+j,j}e^{i\langle{p_q(0), \vec{x} \rangle}}.
\end{align} 
Using \ formulas \eqref{lambda06} and \eqref{u06} and estimates \eqref{2.2.34b-3} and \eqref{2.2.35-3}, we obtain (\ref{solution construction 2l>n}) and \eqref{June7}, respectively.
\end{proof}

\begin{lemma}\label{L:2.12} For any sufficiently large $\lambda $, every $A\in \C: |\sigma||A|^2 <k^{\gamma_0-\delta }$, $\lambda=k^{2l}$ and for every
$\vec{\nu}\in\mathcal{B} (\lambda)$, there is a
unique $\varkappa =\varkappa (\lambda , A, \vec{\nu})$ in the
interval $$
I:=[k-k^{-n+1-2\delta},k+k^{-n+1+2\delta}], $$ such that \begin{equation}\label{2.70}
\lambda (\varkappa \vec{\nu},A)=\lambda . \end{equation}
Furthermore, 
\begin{equation} \label{varkappa}
|\varkappa (\lambda , A, \vec{\nu}) - \tilde k| \leq C(||V||_*)\left(k^{2l-n-\delta}+|\sigma ||A|^2\right)k^{-2l+1-2\gamma_{2}}, \ \ \tilde k=(\lambda -\sigma |A|^2)^{1/2l}.\end{equation}
\end{lemma} \begin{proof} 
Taking into account  \eqref{formulaB} and using   formulas \eqref{2.2.34b-1}, \eqref{2.2.34b-2}  and Implicit Function Theorem,
we prove the lemma.  The proof is completely analogous to that for a linear case. \end{proof}

\begin{theorem} \label{iso} \begin{enumerate} \item For any sufficiently
large $\lambda $ and every $A\in \C: |\sigma||A|^2 <k^{\gamma_0-\delta }$, the set $\mathcal{D}(\lambda , A)$, defined by \eqref{isoset} is a distorted
sphere with holes; it can be described by the formula
\begin{equation}  {\mathcal D}_{}(\lambda, A)=\{\vec k:\vec
k=\varkappa _{}(\lambda, A,\vec{\nu})\vec{\nu},
\ \vec{\nu} \in {\mathcal B}_{}(\lambda)\},\label{May20} \end{equation} where $
\varkappa (\lambda, A, \vec \nu)=\tilde k+h (\lambda, A, \vec \nu)$ and $h (\lambda, A, \vec \nu)$ obeys the
inequalities
\begin{equation}\label{2.75}
|h|<C(||V||_*)\left(k^{2l-n-\delta}+|\sigma ||A|^2\right)k^{-2l+1-2\gamma_{2}}<C(||V||_*)k^{-\gamma _1},
\end{equation}
with $\gamma _1=4l-2-\beta(n-1)-\delta >0$,
\begin{equation}\label{2.75'}\left|\nabla _{\vec \nu }h\right| <
C(||V||_*)k^{-\gamma_1+n-1+2\delta}=C(||V||_*)k^{-2\gamma_2+\delta}.
\end{equation}
\item The measure of $\mathcal{B}(\lambda)\subset S_{n-1}$ satisfies the
estimate
\begin{equation}\label{theta1}
L\left(\mathcal{B}\right)=\omega _{n-1}(1+O(k^{-\delta })).
\end{equation}


\item The surface $\mathcal{D}(\lambda ,A)$ has the measure that is
asymptotically close to that of the whole sphere of the radius $k$ in the sense that
\begin{equation}\label{2.77}
\bigl |\mathcal{D}(\lambda,A )\bigr|\underset{\lambda \rightarrow
\infty}{=}\omega _{n-1}k^{n-1}\bigl(1+O(k^{-\delta})\bigr).
\end{equation}
\end{enumerate}
\end{theorem}

\begin{proof} The proof is based on Implicit Function Theorem. It is completely analogous to
Lemma 2.11 in \cite{K97}.\end{proof}

\


\begin{thebibliography}{99}
























\bibitem[1]{K97} Yulia E. Karpeshina, {\it Perturbation Theory for the Schr\"odinger Operator with a Periodic Potential}, Springer, 1997.

\bibitem[2] {KS17}  Yulia Karpeshina, Seonguk Kim, {\it Solutions of Nonlinear Polyharmonic Equation with Periodic Potentials}, Memorial Volume in The Series Operator Theory: Advances and Applications
(BirkhÃduser), arXiv:1707.01872.


\bibitem[3]{KS02} V. V. Konotop and M. Salerno, {\it Modulational instability in Bose-Einstein condensates in optical lattices}, Phys. Rev. A 65, 021602, 2002.



\bibitem[4]{LO03} Pearl J.Y. Louis, Elena A. Ostrovskaya, Craig M. Savage and Yuri S. Kivshar, {\it Bose-Einstein Condensates in Optical Lattices: Band-Gap Structure and Solitons}, Phys. Rev. A 67, 013602, 2003.


\bibitem[5]{PS08} C. J. Pethick, H. Smith, {\it Bose-Einstein Condensation in Dilute Gases}, Cambridge, 2008.


\bibitem[6]{YB13} A. V. Yulin, Yu. V. Bludov, V. V. Konotop, V. Kuzmiak, and M. Salerno, {\it Superfluidity breakdown of periodic matter waves in quasi-one-dimensional annular traps via resonant scattering with moving defects}, Phys. Rev. A 87, 033625 -- Published 25 March 2013.

\bibitem[7]{YD03} Alexey V. Yulin and Dmitry V. Skryabin, {\it Out-of-gap Bose-Einstein solitons in optical lattices}, Phys. Rev. A 67, 023611, 2003.
\end{thebibliography}
\end{document}